\documentclass[aps,twocolumn,superscriptaddress,nofootinbib,a4paper,longbibliography]{revtex4-2}

\usepackage{graphicx}
\usepackage{epsfig}
\usepackage{amsfonts}
\usepackage{amsmath}
\usepackage{amssymb}
\usepackage{amsthm}

\usepackage{color}
\usepackage[colorlinks=true,linkcolor=blue,citecolor=magenta,urlcolor=blue]{hyperref}
\usepackage[utf8]{inputenc}
\usepackage[english]{babel}
\usepackage[T1]{fontenc}
\usepackage{CJKutf8}
\usepackage{multirow}
\usepackage{subcaption}
\usepackage{amsmath,amssymb,amsthm}
\usepackage{times,ragged2e,soul}
\usepackage[dvipsnames]{xcolor}
\newcommand{\unit}{1\!\!1}

\newcommand\scalemath[2]{\scalebox{#1}{\mbox{\ensuremath{\displaystyle #2}}}}
\usepackage{tikz}

\usepackage{amsmath,amssymb,amsthm}
\usepackage{amsmath,blkarray}
\usepackage{easybmat}
\usepackage{braket}
\usepackage{color}
\usepackage{mathtools}
\usepackage{slashbox}
\usepackage{bbm}
\usepackage{subcaption}
\allowdisplaybreaks
\newcommand{\ketbra}[2]{|#2\rangle\langle#1|}

\newcommand{\be}{\begin{equation}}
\newcommand{\ee}{\end{equation}}
\newcommand{\ba}{\begin{eqnarray}}
\newcommand{\ea}{\end{eqnarray}}

\newcommand{\etal}{{\it{et. al. }}}

\newtheorem{observation}{Observation}
\newtheorem{definition}{Definition}

\newtheorem{result}{Result}
\newcommand*{\Scale}[2][4]{\scalebox{#1}{$#2$}}%

\begin{document}

\title{Contextuality, superlocality and nonclassicality of supernoncontextuality}
\author{Chellasamy Jebarathinam}
 \email{jebarathinam@gmail.com}

\affiliation{Department of Physics and Center for Quantum Information Science, National Cheng Kung University, Tainan 701, Taiwan (R.O.C.)}
 
 \affiliation{Physics Division, National Center for Theoretical Sciences, National Taiwan University, Taipei 106319, Taiwan (R.O.C.)}

\author{R. Srikanth}
\affiliation{Theoretical Sciences Division, Poornaprajna Institute of Scientific Research (PPISR), Bidalur post, Devanahalli, Bengaluru 562164, India}

\begin{abstract}
Contextuality is a fundamental manifestation of nonclassicality, indicating that for certain quantum correlations, sets of jointly measurable variables cannot be pre-assigned values independently of the measurement context.  
In this work, we characterize nonclassical quantum correlation beyond contextuality, in terms of supernoncontextuality, namely the higher-than-quantum hidden-variable(HV) dimensionality required to reproduce the given noncontextual quantum correlations. Thus supernoncontextuality is the contextuality analogue of superlocality.
Specifically, we study the quantum system of two-qubit states in a scenario composed of five contexts that demonstrate contextuality in a state-dependent fashion. For this purpose, we use the framework of boxes, whose behavior is described by a set of probabilities satisfying the no-disturbance conditions.
We first demonstrate that while superlocality is necessary to observe a contextual box, superlocality is not sufficient for contextuality.  On the other hand, a noncontextual superlocal box can be supernoncontextual, but superlocality is not a necessary condition. We then introduce a notion of nonclassicality beyond the standard contextuality, called semi-device-independent contextuality. We study
semi-device-independent contextuality of two-qubit states in the above mentioned scenario and demonstrate how supernoncontextuality implies this nonclassicality. To this end, we propose a criterion and a measure of semi-device-independent contextuality.
\end{abstract}

	\maketitle

\section{Introduction}
The observation of quantum systems admits certain correlations with features that defy classical intuition. Such correlations can come in two setups: Bell scenarios \cite{BCP+14} and Kochen-Specker scenarios \cite{BCG+22}. In a Bell scenario,  outcomes are observed locally on a composite system, whilst on the other hand, in a Kochen-Specker scenario, outcomes are observed in a non-contextual way on a single system that does not have subsystems or a composite system that need not admit a spatial separation between the subsystems. In such scenarios, an outstanding question introduced at the time \cite{Bel66} was whether predefined values to the outcomes could be assigned based on a hidden variable description. Remarkably,   it was shown by Bell that the correlations between the outcomes of quantum measurements in an appropriate setup can violate an inequality based on locality and realism \cite{Bel64}. This phenomenon is referred to as Bell non-locality of quantum correlations \cite{BCP+14}. On the other hand, Kochen and Specker proposed a setup of a certain set of quantum measurements on a quantum system, on which, they showed a logical contradiction with predefined value assignments originating from a hidden variable in a non-contextual way \cite{KS67}. Such yet another remarkable discovery is called Kochen-Specker contextuality of quantum correlations \cite{BCG+22}. Just like Bell inequalities, detecting  Kochen-Specker contextuality based on an inequality satisfied by any non-contextual hidden variable model was also proposed \cite{CFR+08,Cab08,KCB+08}.  

The quantum advantage of quantum correlations due to Bell nonlocality is shown against local-hidden-variable (LHV) models \cite{BHK05}, or, in other words, Bell non-locality of quantum correlations provides certification of 
relevant properties in a Bell scenario in a device-independent way, i.e., without describing the Hilbert-space dimension of the quantum system and the quantum measurements used to demonstrate the phenomenon \cite{AGM06, PAM+10}. This remarkable certification offered by Bell non-locality powers genuinely quantum information protocols such as 
device-independent quantum key distribution  \cite{ABG+07}
device-independent certification of randomness  \cite{PAM+10,AM16}, random access coding using Bell-nonlocal correlations \cite{PZ10,GHH+14,TMB+16} and device-independent certification of quantum devices (such as entangled states, incompatible measurements and quantum channels)  \cite{MY04,SBW+18,SBq20}. Quantum information protocols such as quantum computation, quantum cryptography and quantum communication for which quantum advantage is powered by
Kochen-Specker contextuality have also been explored \cite{AB09,CDV+11,BJB17,SBA17,SHP19}, see  also Sec. VI in Ref. \cite{BCG+22}. Based on contextuality, certification of quantum devices for quantum computation has been shown in a dimension-independent way \cite{BRV+19,BRVC19,IMO+20,SSA20,HXA+22,SJA22,JSS+23}.  A quantum computation device that has such certification in it provides the computational power to have genuine quantum advantage \cite{HWV+14}; otherwise, it is hard to believe that it is truly a  quantum computer \cite{HWA+14}.

However, for practical quantum information processing tasks, quantum advantages shown without Bell nonlocality or quantum contextuality has relevance.  Such quantum advantages are shown with other forms of certification of quantum devices such as certification of entanglement in a one-sided device-independent way \cite{UCN+20}  or fully semi-device-independent way \cite{GBS16}.   Quantum correlations in a Bell scenario have also been studied in a dimension-bounded way, which has relevance for semi-device-independent quantum information processing. In this context, super-locality refers to a
quantum advantage in simulating certain boxes that have a local-hidden-variable model in terms of the dimension of the quantum system against that of the hidden variable \cite{DW15,JAS17,JDS+18}.  In Refs. \cite{DBD+18,JDS23}, the authors studied in a specific scenario superunsteerability, which is an extension of super-locality to steering scenarios \cite{WJD07,UCN+20}. Such studies
provide a better understanding of quantum steering and quantum correlations due to quantum discord \cite{OZ01,HV01}.
Superlocality or superunsteerability also acts as a resource for 
quantum advantage in certain quantum information protocols   \cite{GBS16,JDK+19,JD23,JKC+24}. It would be interesting to extend the concept of superlocality to contextuality scenarios to identify useful new
applications and a useful new understanding of quantum contextuality.

In this work, we characterize these correlations beyond the contextuality of two-qubit states in a scenario that demonstrates contextuality in a state-dependent fashion which is considered by Peres \cite{Per90}. To this end, we adapt the box framework of contextuality to this scenario. We identify that quantum correlations based on superlocality can act as a resource for contextuality. While we observe that superlocality is necessary to demonstrate contextuality in the above mentioned scenario, superlocality of separable states can never lead to contextuality due to the convexity of noncontextual boxes. 
We next introduce the concept of \textit{supernoncontextuality} as the extension of superlocality to contextuality scenarios. Supernoncontextuality
indicates a quantum advantage in using a quantum system of lower Hilbert space dimension to simulate the box over the requirement of higher dimensionality of the noncontextual hidden variable model required to simulate it. 
We then introduce a notion of nonclassicality called \textit{semi-device-independent contextuality} to capture supernoncontextuality that also implies nonclassicality. In the specific scenario that we have considered, 
supernoncontextuality does not necessarily imply semi-device-independent contextuality. To capture and characterize semi-device-independent contextuality
of two-qubit states, we propose a criterion and quantification of semi-device-independent contextuality
of two-qubit states.
 
\section{Preliminaries}

\subsection{Quantum correlation beyond Bell nonlocality}\label{BS}
In a Bell scenario where two spatially separated observers Alice has access to inputs $x$ and observes outputs $a$ and Bob has access to inputs $y$ and observes outputs $b$,  a bipartite box $P$ = $P(ab|xy)$ := $\{ p(ab|xy) \}_{a,x,b,y}$ is the set of joint probability distributions $p(ab|xy)$ for all possible $a$, $x$, $b$, $y$. Such a bipartite box that can be observed in a Bell scenario satisfies the no-signaling (NS) conditions. The single-partite   
box $P(a|x)$ := $\{p(a|x)\}_{a,x}$ of a NS box $P(ab|xy)$ 
is the set of marginal probability distributions $p(a|x)$ for all possible $a$ and $x$; which are given by,
\be
p(a|x)=\sum_b p(ab|xy), \quad \forall a,x,y.
\ee
The single-partite box $P(b|y)$ := $\{p(b|y)\}_{b,y}$ of a NS box $P(ab|xy)$ is the set of marginal
probability distributions 
$p(b|y)$ for all possible $b$ and $y$; which are given by,
\be
p(b|y)=\sum_a p(ab|xy), \quad \forall x,b,y.
\ee

Bell's notion of locality for the NS boxes is defined as follows. A NS box  $P(ab|xy)$ is  Bell-local \cite{Bel64} iff it can be reproduced by an LHV model, 
\begin{equation}\label{LBfs}
p(ab|xy)=\sum_\lambda p(\lambda) p(a|x,\lambda)p(b|y,\lambda) \hspace{0.3cm} \forall a,b,x,y;
\end{equation}
where $\lambda$ denotes shared randomness which occurs with probability $p(\lambda)$; each $p(a|x,\lambda)$ and $p(b|y,\lambda)$ are conditional probabilities. Otherwise, it is Bell nonlocal.
The set of Bell-local boxes which have an LHV model forms a convex polytope called a local 
polytope. Any local box can be written as a convex mixture of the extremal boxes of the local polytope,
\begin{equation}\label{LBfsD}
p(ab|xy)=\sum_i p_i p^{(i)}_D(ab|xy) \hspace{0.3cm} \forall a,b,x,y;
\end{equation}
where $P^{(i)}_D(ab|xy)$ := $\{p^{(i)}_D(ab|xy)\}_{a,x,b,y}$ are local deterministic boxes.

In the context of describing a Bell-local box as in Eqs. (\ref{LBfs}) and (\ref{LBfsD})
using a finite amount of shared randomness, one can identify a quantum advantage with 
superlocality \cite{DW15,JAS17}.
In a bipartite Bell scenario, superlocality  is defined as follows:
\begin{definition}
Suppose we have a quantum state in $\mathbb{C}^{d_A}\otimes\mathbb{C}^{d_B}$
and measurements which produce a local bipartite box $P(ab|xy)$ := $\{ p(ab|xy) \}_{a,x,b,y}$.
Then, superlocality holds iff there is no decomposition of the box in the form,
\begin{equation}
p(ab|xy)=\sum^{d_\lambda-1}_{\lambda=0} p(\lambda) p(a|x, \lambda) p(b|y, \lambda) \hspace{0.3cm} \forall a,x,b,y,
\end{equation}
with the dimension of the shared randomness/hidden variable $d_\lambda\le d^{\rm l}_Q$, with   $d^{\rm l}_Q=\min(d^A_Q, d^B_Q)$ being the minimum local Hilbert space dimension between Alice's and Bob's Hilbert space dimensions. $d^A_Q$ and $d^B_Q$, respectively.  Here $\sum_{\lambda} p(\lambda) = 1$, $p(a|x, \lambda)$ and $p(b|y, \lambda)$ denotes arbitrary probability distributions arising from LHV $\lambda$ ($\lambda$ occurs with probability $p(\lambda)$). 
\end{definition}
To demonstrate superlocality using a bipartite state, a  nonzero two-way discord is necessary. 
Quantum discord captures quantum correlations beyond entanglement \cite{OZ01, HV01}. 
Let $D^{\rightarrow}(\rho_{AB})$ denote quantum discord  as from Alice to Bob, on the other hand, quantum discord  from Bob to Alice 
is denoted as $D^{\leftarrow}(\rho_{AB})$.  Quantum discord from Alice to Bob, $D^{\rightarrow}(\rho_{AB})$,  vanishes for a given $\rho_{AB}$ if and only if (iff) it is a classical-quantum (CQ) state of the form,
\begin{align}
\rho_{\texttt{CQ}}=\sum_i p_i \ket{i} \bra{i}_A \otimes \rho^{(i)}_B,
\end{align}
where  $\{\ket{i}\}$ forms an orthonormal basis on Alice's Hilbert space and $\rho^{(i)}_B$ are any quantum states on Bob's Hilbert space. On the other hand, $D^{\leftarrow}(\rho_{AB})$  vanishes for a given $\rho_{AB}$  iff it is a quantum-classical (QC) state of the form,
\begin{align}
\rho_{\texttt{QC}}=\sum_j p_j  \rho^{(j)}_A  \otimes \ket{j} \bra{j}_B,
\end{align}
where now $\{\ket{j}\}$ forms an orthonormal basis on Bob's Hilbert space and $\rho^{(i)}_A$ are any quantum states on Alice's Hilbert space. A CQ state can have 
one-way discord as measured through $D^{\leftarrow}(\rho_{AB})$, on the other hand, 
a QC state can have one-way discord as measured through $D^{\rightarrow}(\rho_{AB})$.
A quantum-quantum (QQ) state which has a nonzero discord both ways is neither a CQ state nor a QC state.
A classically-correlated (CC) state which has two-way zero discord has  the form,
\begin{align}
\rho_{\texttt{CC}}=\sum_{i,j} p_{ij} \ket{i} \bra{i}_A \otimes \ket{j} \bra{j}_B.
\end{align}
Quantum correlation beyond Bell nonlocality based on superlocality has been studied for 
Bell local states that have a nonzero two-way discord \cite{JAS17}. More specifically, in Ref. \cite{JKC+24}, it was shown that global coherence of two-way discordant states is required to demonstrate superlocality \cite{JKC+24}.

In Ref. \cite{JAS17}, an example of superlocality has been demonstrated in the context of the Bell scenario with two-input and two-output for each side. In this example,  superlocality of the noisy CHSH local box given by 
\begin{equation}
P(ab|xy)       =       \frac{2+(-1)^{a\oplus       b\oplus
    xy}\sqrt{2} V}{8}, \label{chshfam}
\end{equation}
with $ 0<  V \le 1/\sqrt{2}$, was explored.
Such local correlations can be produced by a two-qubit pure entangled state or a two-qubit Werner state having entanglement or a nonzero two-way discord for appropriate local non-commuting measurements. On the other hand, it cannot be reproduced by an LHV model with $d_\lambda=2$ as shown in Ref. \cite{JAS17}. 

In a given scenario, one may discern a strength of superlocality in terms of the number of bits of shared hidden-variable resources required to simulate the given superlocal correlation. Accordingly, the superlocal states in the context of the Bell scenario corresponding to the noisy CHSH box given above have been classified into two classes \cite{DBD+18}: (i) QQ states which demonstrate superlocality with local boxes having minimum hidden variable dimension $3$, and (ii) QQ states which demonstrate superlocality with local boxes having minimum hidden variable dimension $4$.

\subsection{Quantum contextuality}

A contextuality scenario \cite{BCG+22} is described by a set of contexts $C_i$. Each context $C_i$ has a set of observables which are jointly measurable. Such jointly measurable observables are 
mutually commuting, i.e., if the observables $A_1$, $A_2$ and $A_3$ belong to a context,
then they satisfy $[A_1,A_2]=[A_1,A_3]=0$. A contextuality scenario may be described by a compatibility graph which depicts the contexts of the given scenario as in Fig. $1$ of Ref. \cite{XC19}. For the contextuality scenarios that use up to six ideal measurements,
the relationship between the incompatibility of measurements in different contexts and contextuality has been studied in Ref. \cite{XC19}. 

Quantum contextuality has been recently formalized in terms of boxes to characterize better the phenomenon and its applications \cite{GH^314,ACM+18,HJM+23}. Such boxes are called no-disturbance boxes, analogous to no-signaling boxes in Bell scenarios.
Let $P=P(c_i|C_i)$ := $\{ p(c_i|C_i) \}_{c_i,C_i}$ denote a box arising in a contextuality scenario. Here $p(c_i|C_i)$ are the probabilities of jointly observing the outcome set $c_i$  conditioned on measuring the observables in a context $C_i$. 
In a contextuality scenario, the boxes satisfy the no-disturbance (ND) conditions \cite{AQB+13}, analogous to the NS conditions in Bell scenarios. Consider the joint measurements of  $A$ and $B$, and $A$ and $C$ with the joint outcomes denoted by $ab$ and $ac$, respectively. Then the ND condition for these two contexts read as follows: 
\begin{align}
\sum_{b}p(ab|AB)&=\sum_{c}p(ac|AC)=p(a|A),
\end{align}
analogous to the characterization of NS conditions.

Noncontextuality of the ND boxes is defined as follows.
A ND box is noncontextual  iff it has a non-contextual hidden variable model as follows:
\begin{equation}
p(c_i|C_i)=\sum_\lambda p(\lambda) p(c_i|C_i, \lambda)  \hspace{0.3cm} \forall c_i,C_i.
\end{equation}
Here $\sum_{\lambda} p(\lambda) = 1$, $P(c_i|C_i,\lambda)$ := $\{p(c_i|C_i,\lambda)\}_{c_i,C_i}$  denotes an arbitrary noncontextual box for the given $\lambda$ ($\lambda$ occurs with probability $p(\lambda)$). Otherwise, it is contextual.
The set of ND boxes which are noncontextual forms a convex polytope. A non-contextual box can be written as a convex mixture of the deterministic boxes \cite{ACM+18}, 
\be
p(c_i|C_i)=\sum_{j} p_j p^{(j)}_D(c_i|C_i), \quad \forall c_i, C_i,
\ee
where $P^{(j)}_D(c_i|C_i)$ := $\{p^{(j)}_D(c_i|C_i)\}_{c_i,C_i}$ are non-contextual deterministic boxes.

Kochen-Specker contextuality of different set-ups is often investigated through a logical paradox 
or a noncontextuality inequality \cite{BCG+22}. A simple logical paradox with a two-qubit system was addressed by Peres \cite{Per90},  which corresponds to the measurement scenario whose compatibility graph is shown in Fig. \ref{Fig:compatibility1} \cite{CFR+08,XC19}.
Let us choose the observables of the contexts of this scenario as follows:
\begin{align}\label{ObsP}
\begin{split}
&A_0=Z_2 \otimes \unit_2, \quad B_0= \unit_2 \otimes Z_2,  \\
&B_1= \unit_2 \otimes X_2, \quad A_1=X_2 \otimes \unit_2, \\
&D=Z_2 \otimes X_2, \quad E=  X_2  \otimes  Z_2.
\end{split}
\end{align}
Here $\unit_2$, $X_2$, $Y_2$ and $Z_2$ denote the identity operator acting on qubit Hilbert space and the Pauli operators given by
$X_2=\ket{0}\bra{1}+\ket{1}\bra{0}$, $Y_2=-i \ket{0}\bra{1}+i\ket{1}\bra{0}$ and 
$Z_2=\ketbra{0}{0}-\ketbra{1}{1}$, respectively.
Measurements of the observables in Eq. (\ref{ObsP}) 
in the contexts $A_0B_0$, $A_1B_1$ and $DE$ on the two-qubit maximally entangled state given by
\be
\ket{\psi^{\texttt{me}}_2}=\frac{1}{\sqrt{2}}(\ket{00}+\ket{11}),
\ee
implies that the state and observables satisfy the following conditions:
\begin{align}\label{Perl}
\begin{split}
A_0B_0\ket{\psi^{\texttt{me}}_2}&=\ket{\psi^{\texttt{me}}_2}, \\
A_1B_1\ket{\psi^{\texttt{me}}_2}&=\ket{\psi^{\texttt{me}}_2}, \\
DE\ket{\psi^{\texttt{me}}_2}&=-\ket{\psi^{\texttt{me}}_2}.
\end{split}
\end{align}
Note that where in the above equation, we can also write $D$ and $E$ as a product of the two other observables of the scenario as  $D=A_0B_1$ and $E=A_1B_0$ since  $A_0B_1$ and $A_1B_0$ can be jointly measured 
for the above-mentioned observables. Using this, we can write Eq. (\ref{Perl}) in terms of $A_0$, $A_1$, $B_0$ and $B_1$ as follows:
\begin{align}\label{Perl1}
\begin{split}
A_0B_0\ket{\psi^{\texttt{me}}_2}&=\ket{\psi^{\texttt{me}}_2}, \\
A_1B_1\ket{\psi^{\texttt{me}}_2}&=\ket{\psi^{\texttt{me}}_2}, \\
(A_0B_1)(A_1B_0)\ket{\psi^{\texttt{me}}_2}&=-\ket{\psi^{\texttt{me}}_2}.
\end{split}
\end{align}
If we try to assign a non-contextual value $v(\cdot)$ to the outcomes of $A_0$, $A_1$, $B_0$ and $B_1$
satisfying the above conditions, we have
\begin{align}\label{Perv}
\begin{split}
v(A_0)v(B_0)&=1, \\
v(A_0)v(B_1)&=1, \\
v(A_0B_1)v(A_1B_0)&=v(A_0)v(B_1)v(A_1)v(B_0)=-1.
\end{split}
\end{align}
Multiplying both the sides of the three lines in the above equation, we have
$$1=-1,$$ 
where we have used $v^2(A_x)=v^2(B_y)=1$, with $x,y=0,1$. This is a contradiction
since $1 \ne -1$. The above paradox was argued to imply Kochen-Specker contextuality \cite{Per90,Mer90}. 

\begin{figure}[t!]
\begin{center}
\includegraphics[width=5.5cm]{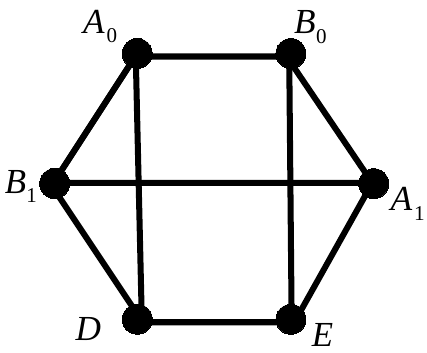}
\end{center}
\caption{Compatibility graph associated with the contextuality scenario of the inequality given by Eq. (\ref{dincI4}). 
\label{Fig:compatibility1}}
\end{figure}

For a box $P(c_i|C_i)=\{p(c_i|C_i)\}$ arising in the scenario in Fig. \ref{Fig:compatibility1}, the probability distributions are denoted as $p(a_0b_0|A_0B_0)$ and $p(a_1b_1|A_1B_1)$
 in case of the contexts $C_0$ and $C_3$, respectively, $p(de|DE)$ in case of the context $C_4$
and $p(a_0b_1d|A_0B_1D)$ and $p(a_1b_0e|A_1B_0E)$ in case of the contexts $C_1$
and $C_2$, respectively. Here $a_l,b_m \in \{0,1\}$, $\forall l, m \in \{0,1\}$ and $d,e\in \{0,1\}$. $D$ and $E$ are a joint measurement acting jointly on two subspaces. For instance, in the case of the set of observables given by Eq. (\ref{ObsP}), the observables  $D$ and $E$ can be realized using a nonlocal gate followed by a Pauli measurement as in Ref. \cite{KZG+09} or implementing joint projectors corresponding to these observables as in Ref. \cite{ARB+09}.

For the scenario whose compatibility graph is Fig. \ref{Fig:compatibility1}, to demonstrate quantum contextuality in an experimental setup, Cabello \etal \cite{CFR+08} considered the following noncontextuality inequality:
\begin{align}\label{dincI4}
\braket{A_0B_0} &+ \braket{A_0B_1D}+\braket{A_1B_0E}+\braket{A_1B_1}  \nonumber \\ 
&- \braket{D E} \le 3,
\end{align}
where each term in the left-hand-side (LHS) is an expectation of the value of the joint measurement of the observables in each context such as $\braket{A_0B_0}=\sum_{a_0,b_0}(-1)^{a_0 \oplus b_0}p(a_0b_0|C_0)$.
Consider the contextual box   $P_P(c_i|C_i)$  arising from the state $\ket{\psi^{\texttt{me}}_2}$ for observables given by Eq. (\ref{ObsP}). This box
has the following distributions:
\be \label{cb2q}
    P_{\rm P}(c_i|C_i) = 
\begin{dcases}
    \frac{1}{2},& \text{if } a_0 + b_0=0\quad \texttt{for} \quad 
    C_0\\
     \frac{1}{4},& \text{if } a_0 + b_1 + d =0\quad \texttt{for} \quad  C_1 \\
      \frac{1}{4},& \text{if } a_1 + b_0 + e =0\quad \texttt{for} \quad  C_2 \\
    \frac{1}{2},& \text{if } a_1 + b_1=0\quad \texttt{for} \quad  C_3 \\
      \frac{1}{2},& \text{if } d + e =1 \quad \texttt{for}  \quad C_4 \\
     0,              &  \texttt{otherwise}. \\ 
\end{dcases}
\ee
We call this box  Peres box.
The maximal violation $5$ of the inequality (\ref{dincI4}) is achieved by the Peres box ($P_{\rm P}(c_i|C_i)$ ).
The Peres box (\ref{cb2q}), which is an extremal contextual box, is given by (according to the matrix notation provided in  \ref{matrixNota}): \\
\begin{align}
P_{\rm P}(c_i|C_i)&=
\begin{pmatrix}
\frac{1}{2} & 0 & 0 & \frac{1}{2} \\ \\[0.05cm]
\frac{1}{4} & 0 & 0 & \frac{1}{4} & 0 & \frac{1}{4} & \frac{1}{4} & 0\\ \\[0.05cm]
\frac{1}{4} & 0 & 0 & \frac{1}{4} & 0 & \frac{1}{4} & \frac{1}{4} & 0\\ \\[0.05cm]
\frac{1}{2} & 0 & 0 & \frac{1}{2} \\ \\[0.05cm]
0 & \frac{1}{2} & \frac{1}{2} & 0  \\
\end{pmatrix}
\end{align}

Based on contextual boxes that imply logical paradoxes, a notion called ``strong contextuality" was formalized, see Sec. IV. A. 4 in Ref. \cite{BCG+22}. Such a notion is characterized as having a measure of 
contextuality called the contextual fraction, $0 \le 1-\alpha \le 1$,  to take its maximal value $1$ \cite{ABM17}. 
The contextual fraction is obtained from the noncontextual fraction or noncontextual content \cite{AB11, ADL+12}, $\alpha$, 
analogous to the quantity ``local fraction" in Bell nonlocality \cite{EPR92} in the context of no-signaling boxes \cite{BCA+11}.  

We consider the measure of contextuality called ``contextuality cost" \cite{GH^314}, related to the above measure. It is defined as 
\begin{align}
{\rm C}(P)&:=\texttt{inf}\{p \in [0,1]: P= pP_C +(1-p)P_{NC} \},
\end{align}
where infimum is the overall possible decomposition of the box into a mixture of some noncontextual box and some contextual box which may be quantum or non-quantum.
The Popescu-Rohrlich box \cite{PR94} is an example of a nonquantum box \cite{AQB+13} that has ${\rm C}(P)=1$ implying maximal contextuality. On the other hand, the Peres box (\ref{cb2q}) also has ${\rm C}(P)=1$  while it is quantum. 

\section{Contextuality in the Peres' scenario}
We analyse in which range a one-parameter family of correlations arising in the scenario of Fig. \ref{Fig:compatibility1} is contextual. This family of correlations is the noisy Peres box $P_{\rm nP}(c_i|C_i)$, given by
\begin{align}\label{wnP}
P_{\rm nP}(c_i|C_i)&=W P_{\rm P}(c_i|C_i)+ (1-W) P_{\rm N}(c_i|C_i), 
\end{align}
with $0 \le W \le 1$, where $P_{\rm P}(c_i|C_i)$ is the Peres box given by Eq. (\ref{cb2q}) and $P_{\rm N}(c_i|C_i)$ is a noncontextual box given by
\be \label{wn}
    P_{\rm N}(c_i|C_i) = 
\begin{dcases}
    \frac{1}{4},&  \forall c_i \quad \texttt{for}  \quad C_0,C_3,C_4\\
    \frac{1}{4},& a_0+b_1+e=0  \quad  \texttt{for} \quad C_1 \\
    \frac{1}{4},&  a_1+b_0+e=0  \quad  \texttt{for} \quad C_2 \\
    0,              &  \texttt{otherwise}. \\ 
\end{dcases}
\ee
The noisy Peres box can be produced quantum-mechanically 
by  two-qubit Werner state given by
\begin{equation}\label{WerFam}
\rho_W=W \ketbra{\psi^{\texttt{me}}_2}{\psi^{\texttt{me}}_2}+ (1-W) \frac{\unit_4}{4},
\end{equation}
for the observables given by Eq. (\ref{ObsP}). The one-parameter family of states is entangled iff $W>1/3$ \cite{Wer89} and has a nonzero quantum discord for any $W>0$ \cite{OZ01,HV01}.

The noisy Peres box violates the noncontextuality inequality 
(\ref{dincI4}) for $W>1/3$, implying contextuality in the range 
$ 1/3 < W \le 1$. On the other hand, for 
$0 \le  W \le 1/3 $, it has a noncontextual hidden variable model. 
This has been checked by finding out that it can be written as a convex mixture of the deterministic boxes in  \ref{DBPscen} in this range.
For instance, consider the noisy Peres box (\ref{wnP}) for $W=1/3$ (adapting the notation in  \ref{matrixNota}) given by
\begin{align} \label{nPbncE}
P^{W=1/3}_{\rm nP}(c_i|C_i)
&:=
\begin{pmatrix}
\frac{1}{3} & \frac{1}{6} & \frac{1}{6} & \frac{1}{3} \\ \\[0.05cm]
\frac{1}{4} & 0 & 0 & \frac{1}{4} & 0 & \frac{1}{4} & \frac{1}{4} & 0 \\ \\[0.05cm]
\frac{1}{4} & 0 & 0 & \frac{1}{4} & 0 & \frac{1}{4} & \frac{1}{4} & 0 \\ \\[0.05cm]
\frac{1}{3} & \frac{1}{6} & \frac{1}{6} & \frac{1}{3}  \\\\[0.05cm]
\frac{1}{6} & \frac{1}{3} & \frac{1}{3} & \frac{1}{6} \\
\end{pmatrix}.
\end{align}
It has been checked that the above box can be written as a convex mixture of the deterministic boxes in  \ref{DBPscen} (see  \ref{CPnPbncE} for the proof). On the other hand, for $W=0$,
the box can be written as the uniform mixture of four noncontextual deterministic boxes as  follows:
\begin{align}
P_{\rm N}&=\frac{1}{4}\Big(P^{(0000)(00)}_D+P^{(0110)(10)}_D+P^{(0111)(01)}_D  \\ \nonumber
&+P^{(0001)(11)}_D\Big), 
\end{align}
where the deterministic boxes are given in  \ref{DBPscen}.
Therefore, for 
$0 \le  W \le 1/3 $, the box is noncontextual.
From the above considerations, we can now state the following observation:
\begin{observation}
    The noisy Peres box $P_{\rm nP}(c_i|C_i)$ given by Eq. (\ref{wnP}) is contextual for  $\frac{1}{3} < W \le 1$ and has the contextuality cost ${\rm C}(P)=\min\{0,3W-1\}$.
\end{observation}

The set of nonsuperlocal boxes forms a \textit{nonconvex} subset of set of all boxes in the given Bell scenario \cite{DW15}. Whereas the set of noncontextual boxes forms a \textit{convex} subset of set of all boxes in the given contextuality scenario \cite{GHH+14,ABM17}. It then follows the the set of noncontextual boxes in the Peres' scenario convexifies the set of    nonsuperlocal boxes in the scenario  with the four contexts 
$\{A_0B_0, A_0B_1, A_1B_0, A_1B_1\}$. From this, we can conclude the following.
\begin{observation}
  Entanglement of superlocal two-qubit states is necessary to demonstrate contextuality  in the Peres' scenario.  
\end{observation}
\begin{proof}
 If there exist  separable two-qubit states which can demonstrate contextuality in the Peres' scenario, then it will imply that the set of noncontextual boxes is not convex. Therefore, all separable states cannot be used to demonstrate contextuality in the Peres' scenario.
 \end{proof}
Whether every two-qubit entangled state is contextual in the Peres' scenario remains to be studied.

From the above observations, we have now obtained the following result.
\begin{result}
Superlocality of two-qubit states is necessary  for producing a contextual box in the state-dependent scenario of Peres. At the same time, superlocality is not sufficient for contextuality since any separable two-qubit state always imply a noncontextual box.  
\end{result}

\section{Nonclassicality of supernoncontextuality}
Any quantum advantage with contextuality is shown against all possible 
noncontextual hidden variable models. Therefore, such quantum advantage is independent of dimension of the quantum systems used to observe contextuality. 
Here we introduce semi-device-independent contextuality as an extension of semi-device-independent nonlocality in Bell scenarios \cite{JKC+24} to contextuality scenarios. The motivation to introduce this notion beyond the standard contextuality is to characterize the nonclassicality of noncontextual boxes that have superlocality. Such noncontextual boxes have supernoncontextuality.

As in the Bell scenarios, in the context of observing a non-contextual box using a finite amount of noncontextual hidden variable \cite{HRA07}, a quantum advantage can be identified in terms of 
supernoncontextuality.
In  a contextuality scenario, supernoncontextuality is defined as follows:
\begin{definition}
Suppose we have a quantum state in $\mathbb{C}^{d^{\rm g}}_{\rm Q}$, with global Hilbert space dimension $d^{\rm g}_{\rm Q}$,
and a set of contexts $\{C_i\}$ which produce a non-contextual box $P(c_i|C_i)$ := $\{ p(c_i|C_i) \}_{c_i,C_i}$.
Then, supernoncontextuality holds iff there is no decomposition of the box in the form,
\begin{equation}
p(c_i|C_i)=\sum^{d_\lambda-1}_{\lambda=0} p(\lambda) p(c_i|C_i, \lambda)  \hspace{0.3cm} \forall c,C_i,
\end{equation}
with dimension of the hidden variable $d_\lambda\le$ $d^{\rm g}_{\rm Q}$.  Here $\sum_{\lambda} p(\lambda) = 1$, $P(c_i|C_i,\lambda)$ := $\{p(c_i|C_i,\lambda)\}_{c_i,C_i}$  denotes an arbitrary non-contextual boxes for the given $\lambda$ ($\lambda$ occurs with probability $p(\lambda)$).
\end{definition}

In the following, we present examples of supernoncontextuality in the Peres' scenario.


\subsection{Supernoncontextuality in the Peres' scenario} 
We start by providing an example of supernoncontextuality 
with the noisy Peres box (\ref{wnP}) for $W=1/3$. 
For this box, from the  decomposition of the box in terms of the deterministic noncontextual boxes as given by Eq. (\ref{DBDnP}), we have obtained the following noncontextual hidden variable model:
\begin{align}
P^{W=1/3}_{\rm nP}(c_i|C_i)&=\sum^{d_\lambda-1}_{\lambda=0} p(\lambda) P(c_i|C_i, \lambda),
\end{align}
with the dimension of the hidden variable $d_\lambda =16$.  Here $\sum_{\lambda} p(\lambda) = 1$, $P(c|C_i,\lambda)$ = $\{p(c|C_i,\lambda)\}_{c,C_i}$  denotes a  non-contextual box for the given $\lambda$ as in Eq. (\ref{DBDnP}), respectively ($\lambda$ occurs with probability $p(\lambda)=1/8$ for $\lambda=0,1,14,15$ and $p(\lambda)=1/24$, otherwise).  It has been checked that 
the simulation of the box using the noncontextual deterministic boxes requires 
the above-mentioned $16$ deterministic boxes. This implies that the box has the minimal $d_\lambda=16$. Therefore, we have obtained the following observation.
\begin{observation}
    The noisy Peres box (\ref{wnP}) for $W=1/3$ is supernoncontextual.  In other words, the box requires a hidden variable of dimension $d_\lambda = 16$ to be simulated using a noncontextual hidden variable model, 
on the other hand, it can be simulated quantum mechanically using a quantum system with global Hilbert space dimension $d^g_Q=4$.
\end{observation}

We next provide  example of supernoncontextuality with the states that can be used to demonstrate contextucality. Consider the maximally entangled state $\ket{\psi^{\texttt{me}}_2}$  that can produce the Peres box. To exhibit supernoncontextuality with this state, we consider the following choice of observables:
\begin{align}\label{Obssn}
\begin{split}
&A_0=Z_2 \otimes \unit_2, \quad B_0= \unit_2 \otimes Z_2,  \\
&B_1= \unit_2 \otimes X_2, \quad A_1=X_2 \otimes \unit_2, \\
&D= \unit_2 \otimes X_2 , \quad E=  X_2  \otimes  \unit_2.
\end{split}
\end{align}
In the above choice of observables, we do not invoke a nonlocal observable for $D$ and $E$
as used for the demonstration of contextuality. This also creates the above choice of observables to have less incompatibility than in the case of the observables (\ref{ObsP}) as there is no incompatibility between $D$ and $A_1$ and $E$ and $B_1$ in Eq. (\ref{Obssn}).

For the observables given above, the box produced from the two-qubit state $\ket{\psi^{\texttt{me}}_2}$ is noncontextual and is  given by
\begin{align}\label{SncMe}
P_{\rm snc}(c_i|C_i)&=
\begin{pmatrix}
\frac{1}{2} & 0 & 0 & \frac{1}{2} \\ \\[0.05cm]
\frac{1}{8} & \frac{1}{8} & \frac{1}{8} & \frac{1}{8} & \frac{1}{8} & \frac{1}{8} & \frac{1}{8} & \frac{1}{8}\\ \\[0.05cm]
\frac{1}{8} & \frac{1}{8} & \frac{1}{8} & \frac{1}{8} & \frac{1}{8} & \frac{1}{8} & \frac{1}{8} & \frac{1}{8}\\ \\[0.05cm]
\frac{1}{2} & 0 & 0 & \frac{1}{2} \\ \\[0.05cm]
\frac{1}{2} & 0 & 0 & \frac{1}{2}  \\
\end{pmatrix}
\end{align}
The  box given above has the following noncontextual deterministic model:
\begin{align}\label{sncsl}
&P_{\rm snc} \nonumber \\
&=\frac{1}{8}\Big(P^{(0000)(00)}_D+P^{(0101)(00)}_D+P^{(1010)(00)}_D  \\ \nonumber
&+P^{(1111)(00)}_D +P^{(0000)(11)}_D+P^{(0101)(11)}_D  \\ \nonumber
&+P^{(1010)(11)}_D+P^{(1111)(11)}_D\Big). 
\end{align}
The simulation of the box using the deterministic boxes requires the $8$ deterministic boxes as appearing in the above decomposition, which implies the following noncontextual model:
\begin{align}
P(c_i|C_i)&=\sum^{d_\lambda-1}_{\lambda=0} p(\lambda) P(c_i|C_i, \lambda),
\end{align}
with the dimension of the hidden variable $d_\lambda =8$. Here $P(c_i|C_i,\lambda)=\{p(c_i|C_i,\lambda)\}_{c_i,C_i}$ are the deterministic boxes appearing in Eq. (\ref{sncsl}), respectively. From the above noncontexual model, it  follows that the box is supernoncontextual since $d_\lambda=8>d_Q$ and this value of $d_\lambda$ is the minimal overall possible noncontextual hidden variable models. 
We have obtained the following observation.
\begin{observation}
    Supernoncontextuality of certain noncontextual and superlocal states as in the case of Werner state family occurs due to the contextual resource of measurements which has full incompatibility between all contexts. 
    On the other hand, supernoncontextuality of certain contextual states  in the Peres' scenario occurs due to superlocality and not requiring full incompatibility between measurements in all contexts as in the case of demonstration of contextuality.  
\end{observation}
In fact, supernoncontextuality also occurs for nonsuperlocal states as we have demonstrated in  \ref{Exsncwsl}. We have obtained the following observation.
\begin{observation}
    While superlocality implies supernoncontextuality  in the Peres' scenario,
    superlocality is not necessary for supernoncontextuality in this scenario.
\end{observation}

\subsection{Semi-device-independent contextuality}
The demonstration of standard contextuality in any scenario implies the presence of incompatibility of measurements between all different contexts. This happens in a dimension-independent way as the standard contextuality is shown against all possible noncontextual hidden variable models. We now introduce the notion of semi-device-independent contextuality as follows.
\begin{definition}
    Suppose we have a quantum state in $\mathbb{C}^{d^{\rm g}_{\rm Q}}$
and a set of measurements of the contexts $\{C_i\}$, which produce a  box $P(c_i|C_i)$.
Then, the box implies the presence of  semi-device-independent contextuality  iff there is no decomposition of the box in the form,
\begin{equation}
p(c_i|C_i)=\sum^{d_\lambda-1}_{\lambda=0} p(\lambda) p(c_i|C_i, \lambda)  \hspace{0.3cm} \forall c_i,C_i,
\end{equation}
with the dimension of the hidden variable $d_\lambda\le$ $d^{\rm g}_{\rm Q}$, provided that the box detects incompatibility of measurements between  all different contexts. 
\end{definition}

We now study semi-device-independent contextuality of supernoncontextual boxes in the Peres' scenario.
We have already seen that there are noncontextual boxes in the Peres' scenario whose supernoncontextuality does not require the full incompatibility of measurements between different contexts or does not even require superlocality. Therefore, these supernoncontextual boxes do not exhibit semi-device-independent contextuality.
To witness semi-device-independent contextuality of supernoncontextual boxes,
we propose the following criterion. In Ref. \cite{JD23}, a nonlinear witness 
of superlocality was introduced. This witness was defined in terms of conditional probabilities, for our purpose, we adopt this witness
in terms of the covariance,  $cov(A_x,B_y)$, of the observables $A_x$ and $B_y$ given by
 \be
 cov(A_x,B_y)=\braket{A_xB_y} -\braket{A_x}\braket{B_y},
 \ee
 where $\braket{A_xB_y}$ and $\braket{A_x}$, and $\braket{B_y}$
 are joint and marginal expectation values, respectively.
The witness of superlocality in terms of the covariances is then given by
\begin{align}
	Q&=\left|\begin{array}{cc}cov(A_0,B_0) & cov(A_1,B_0)\\ 
		cov(A_0,B_1) & cov(A_1,B_1) \end{array}\right|. \label{QC}
\end{align}
We now state the following result.
\begin{result} \label{secres}
Suppose we have a quantum state in $\mathbb{C}^{4}$
and a set of measurements of the contexts $\{C_i\}$ in the Peres' scenario, which produce a  box $P(c_i|C_i)$.
Then, the box implies the presence of  semi-device-independent contextuality 
if the witness of superlocality $Q$ in Eq. (\ref{QC}) is nonzero and $\braket{A_0B_1D}=\braket{A_1B_0E}=1$, with  $cov(D,E)>0$. 
\end{result}
In the above conditions, witnessing superlocality implies supernoncontextuality as well as incompatibility between measurements in the two context $A_0B_0$ and $A_1B_1$, and, the other condition $\braket{A_0B_1D}=\braket{A_1B_0E}=1$, with  $cov(D,E)>0$ has been invoked to imply incompatibility between the measurements in the other contexts so that  semi-device-independent contextuality 
is witnessed.

We proceed to apply the above criterion of semi-device-independent contextuality to some supernoncontextual boxes. Note that the supernoncontextual boxes given by Eqs. (\ref{SncMe}), (\ref{sncbd}) and  (\ref{cc2qhd}) do not satisfy the criterion for semi-device-independent contextuality in Result. \ref{secres}. 
For the noisy Peres box $P_{\rm nP}(c_i|C_i)$ (\ref{wnP}) which is noncontextual in the range $0 \le W  \le \frac{1}{3}$, the witness of superlocality $Q$ in Eq. (\ref{QC}) takes the value $Q=W^2$ which is nonzero for any $W>0$ and it also satisfies
the other condition in Result. \ref{secres} for any $W>0$. Therefore, the noisy Peres box $P_{\rm nP}(c_i|C_i)$ has semi-device-independent contextuality for any $W>0$.

The noisy Peres box (\ref{wnP}) for $W=1/3$ has superlocality with $d_\lambda=4$. This can be observed by noting that the marginal of this box 
$\{p(a_xb_y|A_xB_y)\}$ has superlocality with $d_\lambda=4$ \cite{DBD+18}. In fact,
the separable Werner state $\rho_W$ with $W=\frac{1}{3}$ that gives rise to the supernoncontextual box can be written as a convex mixture of product states with cardinality/rank $4$ as follows. 
 To express the separable Werner state $\rho_W$ with $W=\frac{1}{3}$ as a convex mixture of minimal product states, we define the state $\ket{\theta,\phi}$ as 
    \begin{equation} \label{eq:desc_sep_1}
    \ket{\theta,\phi} := \cos\left(\frac{\theta}{2}\right) \ket{0}
                        + \exp(i \phi)\sin\left(\frac{\theta}{2}\right) \ket{1} \,,
    \end{equation}
which is an arbitrary pure state in $\mathbb{C}^2$
 and the set $\{\mathcal{Z}\}$ as ${\mathcal{Z}=\{\ket{0,0},\ket{\theta^*,0},
\ket{\theta^*,\frac{2\pi}{3}},\ket{\theta^*,\frac{4\pi}{3}}\}}$,
with ${\theta^*=\arccos(-\frac{1}{3})}$. Then the above mentioned separable Werner state can be decomposed as follows \cite{BPP15}:
    \begin{equation} \label{eq:desc_sep_2}
     \rho^{W=1/3}_{W} = \frac{1}{4} \sum_{k=1}^4{Z_k\otimes Z_k} \,,
    \end{equation}
with ${Z_k=\ket{z_k}\bra{z_k}}$ and ${\ket{z_k}\in\mathcal{Z}}$.
Thus, the state leads to superlocality with $d_\lambda=4$.

 We next consider a separable two-qubit state that has superlocality with $d_\lambda=3$. To this end,  we
 define the set ${\mathcal{W}=\{\ket{0,0},\ket{\frac{2\pi}{3},0},\ket{\frac{2\pi}{3},\pi}\}}$ with the state given by Eq. (\ref{eq:desc_sep_1}).
Then consider the separable state that can be expressed as 
    \begin{equation} \label{eq:rk3_max_opt}
    \rho^{\texttt{rank}-3}_{\rm QQ} = \frac{1}{3}\sum_{k=1}^3{W_k\otimes W_k} \,,
    \end{equation}
with ${W_k=\ket{w_k}\bra{w_k}}$ and ${\ket{w_k}\in\mathcal{W}}$.
The above state is equivalent to the state 
\begin{align}\label{rank3sep}
\rho_{\rm Sl}^{d_\lambda=3} &= \frac{1}{4} \Big( |00\rangle \langle 00| + |++ \rangle \langle ++|  \nonumber \\ 
&+|11\rangle \langle 11| + |-- \rangle \langle --|
\Big),   
\end{align}
up to local unitary transformation \cite{BPP15}. The above state which belongs to the Bell-diagonal family is superlocal \cite{JDK+19} and has superlocality with $d_\lambda=3$ since the cardinality of the separable decomposition is $3$.
For the measurements given by Eq. (\ref{ObsP}), the box arising from the superlocal state (\ref{rank3sep}) satisfies the conditions in Result. \ref{secres} except $cov(D,E)>0$. Hence this supernoncontextual box is not semi-device-independent contextual.
 
We  are now going to consider a supernoncontextual box arising from another rank-$3$ separable state that has superlocality with $d_\lambda=3$, but has semi-device-independent contextuality. This state is given by
\begin{align}\label{rank3}
\sigma_{\rm Sl}^{d_\lambda=3}&=\frac{1}{3} \Big( \ketbra{00}{00} +\ketbra{++}{++}   \nonumber \\  
&+\ketbra{+'+'}{+'+'}  \Big).
\end{align}
Here $\ket{+'}$ is the $+$ eigenstate of the Pauli observables $Y_2$. 
Consider the following box arising from the rank-$3$ state (\ref{rank3}) for the observables given by Eq. (\ref{ObsP}):
\begin{align}\label{cbd}
P_{\rm snc}(c_i|C_i)&=
\begin{pmatrix}
\frac{1}{2} & \frac{1}{6} & \frac{1}{6} & \frac{1}{6} \\ \\[0.05cm]
\frac{5}{12} & 0 & 0 & \frac{1}{12} & 0 & \frac{1}{4} & \frac{1}{4} & 0 \\ \\[0.05cm]
\frac{5}{12} & 0 & 0 & \frac{1}{12} & 0 & \frac{1}{4} & \frac{1}{4} & 0  \\ \\[0.05cm]
\frac{1}{2} & \frac{1}{6} & \frac{1}{6} & \frac{1}{6}  \\\\[0.05cm]
\frac{1}{6} & \frac{1}{3} & \frac{1}{3} & \frac{1}{6} \\
\end{pmatrix}.
\end{align}
It has been checked that the above box satisfies all conditions in Result. \ref{secres}. Hence it is semi-device-independent contextual.

We  are now going to propose a measure  of semi-device-independent contextuality to characterize the nonclassicality quantitatively. We define this measure, called Peres strength, analogues to the quantification of Bell nonclassicality using Bell strength in Ref. \cite{JAS17}. 
\begin{definition}
    Consider the decomposition of any given box $P(c_i|C_i)$ in the Peres' scenario into a convex mixture of a maximally contextual box such as a Peres box $P_{\rm P}(c_i|C_i)$ as in Eq. (\ref{cb2q}) and a noncontextual box  $P_{\rm nc}(c_i|C_i)$,
    \begin{equation}
        P(c_i|C_i)=p P_P(c_i|C_i) +(1-p) P_{\rm nc}(c_i|C_i).
    \end{equation}
    Peres strength, $\rm PS(P)$, of the given box is then defined as the fraction of the Peres box $P_{\rm P}(c_i|C_i)$ maximized over all possible decomposition of the given box. 
\end{definition}
A nonzero Peres strength, $\rm PS(P)$, of the given supernoncontextual box  that satisfies the conditions in Result. \ref{secres} can be used to quantify semi-device-independent contextuality.

We now calculate Peres strength of two examples of semi-device-independent contextuality. For the noisy Peres box, the decomposition given by Eq. (\ref{wnP}) provides the optimal decomposition to provide  Peres strength as $\rm PS(P_{\rm nP})=W$. On the other hand, for the box (\ref{cbd}), Peres strength is obtained using the  following decomposition of the box:
\begin{equation}
P_{\rm snc}(c_i|C_i)=\frac{1}{3} P_P(c_i|C_i)+ \frac{2}{3} P_{\rm nc}(c_i|C_i),
\end{equation}
where 
\begin{align}\label{ncbd}
P_{\rm nc}(c_i|C_i)&=
\begin{pmatrix}
\frac{1}{2} & \frac{1}{4} & \frac{1}{4} & 0 \\ \\[0.05cm]
\frac{1}{2} & 0 & 0 & 0 & 0 & \frac{1}{4} & \frac{1}{4} & 0 \\ \\[0.05cm]
\frac{1}{2} & 0 & 0 & 0 & 0 & \frac{1}{4} & \frac{1}{4} & 0  \\ \\[0.05cm]
\frac{1}{2} & \frac{1}{4} & \frac{1}{4} & 0   \\\\[0.05cm]
\frac{1}{4} & \frac{1}{4} & \frac{1}{4} & \frac{1}{4} \\
\end{pmatrix}.
\end{align}
The above decomposition provides the Peres strength of the box as $\rm PS(P_{\rm snc})=\frac{1}{3}$.

 \section{Conclusions}
 In this work, we have been motivated to extend semi-device-independent nonlocality based on superlocality \cite{JKC+24} to contextuality scenarios.
 This is relevant to understand  contextuality differently and to identify the quantum resource for quantum computation without contextuality. To make progress in this context, we have considered the specific scenario that corresponds to the Peres' proof of contextuality \cite{Per90}, which has the five contexts as in Fig. \ref{Fig:compatibility1}. First, we have characterized contextuality in this scenario by adopting the framework of boxes in the contextuality scenarios. We have then found an interesting relationship between superlocality and contextuality. Specifically, we have demonstrated that superlocality 
 in the two-input and two-output Bell scenario is necessary to demonstrate contextuality, while it is not sufficient. 
 
 Superlocality of all separable two-qubit states always implies a noncontextual box in the specific scenario since the set of noncontextual boxes is convex.  We are then interested in the question of whether noncontextual boxes that have superlocality has nonclassicality, just like semi-device-independent nonlocality occurs for Bell-local states. 
 To this end, we have considered supernoncontextuality which is an extension of superlocality to contextuality scenarios. Supernoncontextuality
indicates a quantum advantage in using a quantum system of lower Hilbert space dimension to simulate the box over the requirement of high dimensionality of the hidden 
variable required to simulate it using a noncontextual hidden variable model. 
However such a quantum advantage, in the specific scenario that we have considered does not necessarily imply nonclassicality. We have defined this nonclassicality beyond the standard contextuality as semi-device-independent contexutality. Finally, we have studied how semi-device-independent contexutality occurs for supernoncontextual boxes in the specific scenario by introducing a criterion and quantification of the nonclassicality.

\section*{Acknowledgement}
C. J. would like to thank Dr Ashutosh Rai, Dr Jaskaran Singh, Dr Debarshi Das
and Dr Wei-Min Zhang for useful discussions
and the National Science
and Technology Council (formerly Ministry of Science and
Technology), Taiwan  (Grant No. MOST 111-2811-M-006-040-MY2).
RS acknowledges with thanks the partial support by the Indian Science \& Engineering Research Board (SERB grant CRG/2022/008345).
This work was supported by the National Science and Technology Council, the Ministry of Education (Higher Education Sprout Project NTU-113L104022-1), and the National Center for Theoretical Sciences of Taiwan.

\onecolumngrid
\appendix

\section{Noncontextual deterministic boxes of the scenario in Fig. \ref{Fig:compatibility1}}\label{DBPscen}
Let us denote the ND boxes $P(c_i|C_i)$ of the contextuality 
scenario with the compatibility graph as shown in Fig. \ref{Fig:compatibility1} in the matrix form as follows:
\[
\scalemath{0.76}{
\begin{aligned}\label{matrixNota}
&P(c_i|C_i)=   \\
&\begin{pmatrix}
p(00|A_0B_0) & p(01|A_0B_0) & p(10|A_0B_0) & p(11|A_0B_0) \\ \\[0.05cm]
p(000|A_0B_1D) & p(010|A_0B_1D) & p(100|A_0B_1D) & p(110|A_0B_1D) &p(001|A_0B_1D) & p(011|A_0B_1D) & p(101|A_0B_1D) & p(111|A_0B_1D)  \\ \\[0.05cm]
p(000|A_1B_0E) & p(010|A_1B_0E) & p(100|A_1B_0E) & p(110|A_1B_0E) &p(001|A_1B_0E) & p(011|A_1B_0E) & p(101|A_1B_0E) & p(111|A_1B_0E) \\ \\[0.05cm]
p(00|A_1B_1) & p(01|A_1B_1) & p(10|A_1B_1) & p(11|A_1B_1) \\ \\[0.05cm]
p(00|DE) & p(01|DE) & p(10|DE) & p(11|DE)  \\ 
\end{pmatrix}
\end{aligned}
}\]

The scenario has $64$ extremal noncontextual boxes which are deterministic boxes denoted as $P^{(\alpha\beta\gamma\epsilon)(de)}_D$, with $\alpha, \beta, \gamma, \epsilon \in \{0,1\}$.  $P^{\alpha\beta\gamma\epsilon}_D(a_xb_y|A_xB_y)$  denote 
the marginal deterministic boxes of $P^{(\alpha\beta\gamma\epsilon)(de)}_D$.  
$P^{\alpha\beta\gamma\epsilon}_D(a_xb_y|A_xB_y)$ are given by
\be \label{LDB}
    P_{D}^{\alpha \beta \gamma \epsilon} (a_xb_y|A_xB_y) = 
\begin{dcases}
    1,& \text{if } a_x = \alpha x \oplus \beta, b_y = \gamma y \oplus \epsilon \\
    0,              & \text{otherwise},
\end{dcases}
\ee
in matrix form, the above boxes are given by
\begin{align*}
 P_{D}^{0000}=
\begin{pmatrix}
1 & 0 & 0 & 0 \\ 
1 & 0 & 0 & 0 \\ 
1 & 0 & 0 & 0 \\
1 & 0 & 0 & 0 \\ 
\end{pmatrix};~~~
 P_{D}^{0001}=
\begin{pmatrix}
0 & 1 & 0 & 0 \\ 
0 & 1 & 0 & 0 \\ 
0 & 1 & 0 & 0 \\ 
0 & 1 & 0 & 0 \\
\end{pmatrix};~~~
 P_{D}^{0010}=
\begin{pmatrix}
1 & 0 & 0 & 0 \\
0 & 1 & 0 & 0 \\   
1 & 0 & 0 & 0 \\ 
0 & 1 & 0 & 0 \\
\end{pmatrix};~~~
 P_{D}^{0011}=
\begin{pmatrix}
0 & 1 & 0 & 0 \\ 
1 & 0 & 0 & 0 \\ 
0 & 1 & 0 & 0 \\ 
1 & 0 & 0 & 0 \\
\end{pmatrix};\\\\
 P_{D}^{0100}=
\begin{pmatrix}
0 & 0 & 1 & 0 \\ 
0 & 0 & 1 & 0 \\ 
0 & 0 & 1 & 0 \\ 
0 & 0 & 1 & 0 \\
\end{pmatrix};~~~
 P_{D}^{0101}=
\begin{pmatrix}
0 & 0 & 0 & 1 \\  
0 & 0 & 0 & 1 \\
0 & 0 & 0 & 1 \\ 
0 & 0 & 0 & 1 \\
\end{pmatrix};~~~
 P_{D}^{0110}=
\begin{pmatrix}
0 & 0 & 1 & 0 \\ 
0 & 0 & 0 & 1 \\ 
0 & 0 & 1 & 0 \\ 
0 & 0 & 0 & 1 \\
\end{pmatrix};~~~
 P_{D}^{0111}=
\begin{pmatrix}
0 & 0 & 0 & 1 \\ 
0 & 0 & 1 & 0 \\ 
0 & 0 & 0 & 1 \\ 
0 & 0 & 1 & 0 \\
\end{pmatrix};\\\\
 P_{D}^{1000}=
\begin{pmatrix}
1 & 0 & 0 & 0 \\
1 & 0 & 0 & 0 \\
0 & 0 & 1 & 0 \\ 
0 & 0 & 1 & 0 \\
\end{pmatrix};~~~
 P_{D}^{1001}=
\begin{pmatrix}
0 & 1 & 0 & 0 \\ 
0 & 1 & 0 & 0 \\ 
0 & 0 & 0 & 1 \\
0 & 0 & 0 & 1 \\
\end{pmatrix};~~~
 P_{D}^{1010}=
\begin{pmatrix}
1 & 0 & 0 & 0 \\ 
0 & 1 & 0 & 0 \\ 
0 & 0 & 1 & 0 \\ 
0 & 0 & 0 & 1 \\
\end{pmatrix};~~~
 P_{D}^{1011}=
\begin{pmatrix}
0 & 1 & 0 & 0 \\ 
1 & 0 & 0 & 0 \\ 
0 & 0 & 0 & 1 \\ 
0 & 0 & 1 & 0 \\
\end{pmatrix}; \\\\
 P_{D}^{1100}=
\begin{pmatrix}
0 & 0 & 1 & 0 \\
0 & 0 & 1 & 0 \\ 
1 & 0 & 0 & 0 \\ 
1 & 0 & 0 & 0 \\
\end{pmatrix};~~~
 P_{D}^{1101}=
\begin{pmatrix}
0 & 0 & 0 & 1 \\ 
0 & 0 & 0 & 1 \\ 
0 & 1 & 0 & 0 \\
0 & 1 & 0 & 0 \\
\end{pmatrix};~~~
 P_{D}^{1110}=
\begin{pmatrix}
0 & 0 & 1 & 0 \\  
0 & 0 & 0 & 1 \\
1 & 0 & 0 & 0 \\ 
0 & 1 & 0 & 0 \\
\end{pmatrix};~~~
 P_{D}^{1111}=
\begin{pmatrix}
0 & 0 & 0 & 1 \\  
0 & 0 & 1 & 0 \\ 
0 & 1 & 0 & 0 \\
1 & 0 & 0 & 0 \\
\end{pmatrix}
\end{align*}

The  probability $p_D(00|DE)$ takes the value $1$ if $de=00$ in $P^{(\alpha\beta\gamma\epsilon)(de)}_D$.
Using the above mentioned notations, the $16$ deterministic boxes $P^{(\alpha\beta\gamma\epsilon)(00)}_D$ are given in the matrix form as follows:
\[
\scalemath{0.8}{
\begin{aligned}\label{D00}
\begin{split}
\begin{pmatrix}
1 & 0 & 0 & 0 \\
1 & 0 & 0 & 0 & 0 & 0 & 0 & 0 \\
1 & 0 & 0 & 0 & 0 & 0 & 0 & 0 \\ 
1 & 0 & 0 & 0 \\ 
1 & 0 & 0 & 0  \\
\end{pmatrix};~~~
\begin{pmatrix}
0 & 1 & 0 & 0 \\  
0 & 1 & 0 & 0 & 0 & 0 & 0 & 0 \\ 
0 & 1 & 0 & 0 & 0 & 0 & 0 & 0 \\ 
0 & 1 & 0 & 0 \\ 
1& 0 & 0 & 0  \\
\end{pmatrix};~~~
\begin{pmatrix}
1 & 0 & 0 & 0 \\  
0 & 1 & 0 & 0  & 0 & 0 & 0 & 0 \\  
1 & 0 & 0 & 0  & 0 & 0 & 0 & 0 \\  
0 & 1 & 0 & 0 \\ 
1 & 0 & 0 & 0  \\
\end{pmatrix};~~~
\begin{pmatrix}
0 & 1 & 0 & 0 \\ 
1 & 0 & 0 & 0 & 0 & 0 & 0 & 0 \\ 
0 & 1 & 0 & 0 & 0 & 0 & 0 & 0 \\ 
1 & 0 & 0 & 0 \\ 
1& 0 & 0 & 0  \\
\end{pmatrix};\\\\
\begin{pmatrix}
0 & 0 & 1 & 0 \\  
0 & 0 & 1 & 0 & 0 & 0 & 0 & 0 \\ 
0 & 0 & 1 & 0 & 0 & 0 & 0 & 0 \\ 
0 & 0 & 1 & 0 \\  
1 & 0 & 0 & 0  \\
\end{pmatrix};~~~
\begin{pmatrix}
0 & 0 & 0 & 1 \\ 
0 & 0 & 0 & 1 & 0 & 0 & 0 & 0 \\ 
0 & 0 & 0 & 1 & 0 & 0 & 0 & 0 \\ 
0 & 0 & 0 & 1 \\  
1 & 0 & 0 & 0  \\
\end{pmatrix};~~~
\begin{pmatrix}
0 & 0 & 1 & 0 \\ 
0 & 0 & 0 & 1 & 0 & 0 & 0 & 0 \\ 
0 & 0 & 1 & 0 & 0 & 0 & 0 & 0 \\
0 & 0 & 0 & 1 \\ 
1 & 0 & 0 & 0  \\
\end{pmatrix};~~~
\begin{pmatrix}
0 & 0 & 0 & 1 \\ 
0 & 0 & 1 & 0 & 0 & 0 & 0 & 0 \\ 
0 & 0 & 0 & 1 & 0 & 0 & 0 & 0 \\ 
0 & 0 & 1 & 0 \\ 
1 & 0 & 0 & 0  \\
\end{pmatrix};\\\\
\begin{pmatrix}
1 & 0 & 0 & 0 \\
1 & 0 & 0 & 0 & 0 & 0 & 0 & 0 \\
0 & 0 & 1 & 0 & 0 & 0 & 0 & 0 \\
0 & 0 & 1 & 0 \\ 
1 & 0 & 0 & 0  \\
\end{pmatrix};~~~
\begin{pmatrix}  
0 & 1 & 0 & 0 \\ 
0 & 1 & 0 & 0 & 0 & 0 & 0 & 0 \\
0 & 0 & 0 & 1 & 0 & 0 & 0 & 0 \\
0 & 0 & 0 & 1 \\
1& 0 & 0 & 0  \\
\end{pmatrix};~~~
\begin{pmatrix}
1 & 0 & 0 & 0 \\
0 & 1 & 0 & 0 & 0 & 0 & 0 & 0 \\
0 & 0 & 1 & 0 & 0 & 0 & 0 & 0 \\
0 & 0 & 0 & 1 \\
1 & 0 & 0 & 0  \\
\end{pmatrix};~~~
\begin{pmatrix}
0 & 1 & 0 & 0 \\
1 & 0 & 0 & 0 & 0 & 0 & 0 & 0 \\
0 & 0 & 0 & 1 & 0 & 0 & 0 & 0 \\
0 & 0 & 1 & 0 \\
1& 0 & 0 & 0  \\
\end{pmatrix}; \\\\
\begin{pmatrix}
0 & 0 & 1 & 0 \\
0 & 0 & 1 & 0 & 0 & 0 & 0 & 0\\
1 & 0 & 0 & 0 & 0 & 0 & 0 & 0 \\
1 & 0 & 0 & 0 \\
1 & 0 & 0 & 0  \\
\end{pmatrix};~~~
\begin{pmatrix}
0 & 0 & 0 & 1 \\
0 & 0 & 0 & 1 & 0 & 0 & 0 & 0\\
0 & 1 & 0 & 0 & 0 & 0 & 0 & 0 \\
0 & 1 & 0 & 0 \\
1 & 0 & 0 & 0  \\
\end{pmatrix};~~~
\begin{pmatrix}
0 & 0 & 1 & 0 \\
0 & 0 & 0 & 1 & 0 & 0 & 0 & 0 \\
1 & 0 & 0 & 0 & 0 & 0 & 0 & 0 \\
0 & 1 & 0 & 0 \\
1 & 0 & 0 & 0  \\
\end{pmatrix};~~~
\begin{pmatrix}
0 & 0 & 0 & 1 \\
0 & 0 & 1 & 0 & 0 & 0 & 0 & 0 \\
0 & 1 & 0 & 0 & 0 & 0 & 0 & 0 \\
1 & 0 & 0 & 0 \\
1 & 0 & 0 & 0  \\
\end{pmatrix},
\end{split}
\end{aligned}
}\]
respectively, the $16$ deterministic boxes for $de=10$ in $P^{(\alpha\beta\gamma\epsilon)(de)}_D$ are given as follows:
\[
\scalemath{0.8}{
\begin{aligned}\label{D10}
\begin{split}
\begin{pmatrix}
1 & 0 & 0 & 0 \\
0 & 0 & 0 & 0 & 1 & 0 & 0 & 0\\
1 & 0 & 0 & 0 & 0 & 0 & 0 & 0\\
1 & 0 & 0 & 0 \\
0 & 0 & 1 & 0  \\
\end{pmatrix};~~~
\begin{pmatrix}
0 & 1 & 0 & 0 \\
0 & 0 & 0 & 0 & 0 & 1 & 0 & 0\\
0 & 1 & 0 & 0 & 0 & 0 & 0 & 0\\
0 & 1 & 0 & 0 \\
0& 0 & 1 & 0  \\
\end{pmatrix};~~~
\begin{pmatrix}
1 & 0 & 0 & 0 \\
0 & 0 & 0 & 0 & 0 & 1 & 0 & 0 \\
1 & 0 & 0 & 0 & 0 & 0 & 0 & 0 \\
0 & 1 & 0 & 0 \\
0 & 0 & 1 & 0  \\
\end{pmatrix};~~~
\begin{pmatrix}
0 & 1 & 0 & 0 \\
0 & 0 & 0 & 0 & 1 & 0 & 0 & 0\\
0 & 1 & 0 & 0 & 0 & 0 & 0 & 0\\
1 & 0 & 0 & 0 \\
0& 0 & 1 & 0  \\
\end{pmatrix};\\\\
\begin{pmatrix}
0 & 0 & 1 & 0 \\
0 & 0 & 0 & 0 & 0 & 0 & 1 & 0\\
0 & 0 & 1 & 0  & 0 & 0 & 0 & 0\\
0 & 0 & 1 & 0\\
0 & 0 & 1 & 0  \\
\end{pmatrix};~~~
\begin{pmatrix}
0 & 0 & 0 & 1 \\
0 & 0 & 0 & 0 & 0 & 0 & 0 & 1 \\
0 & 0 & 0 & 1  & 0 & 0 & 0 & 0\\
0 & 0 & 0 & 1\\
0 & 0 & 1 & 0  \\
\end{pmatrix};~~~
\begin{pmatrix}
0 & 0 & 1 & 0 \\
0 & 0 & 0 & 0 & 0 & 0 & 0 & 1\\
0 & 0 & 1 & 0 & 0 & 0 & 0 & 0\\
0 & 0 & 0 & 1 \\
0 & 0 & 1 & 0  \\
\end{pmatrix};~~~
\begin{pmatrix}
0 & 0 & 0 & 1 \\
0 & 0 & 0 & 0 & 0 & 0 & 1 & 0\\
0 & 0 & 0 & 1 & 0 & 0 & 0 & 0\\
0 & 0 & 1 & 0 \\
0 & 0 & 1 & 0  \\
\end{pmatrix};\\\\
\begin{pmatrix}
1 & 0 & 0 & 0 \\
0 & 0 & 0 & 0 & 1 & 0 & 0 & 0\\
0 & 0 & 1 & 0 & 0 & 0 & 0 & 0\\
0 & 0 & 1 & 0 \\
0 & 0 & 1 & 0  \\
\end{pmatrix};~~~
\begin{pmatrix}
0 & 1 & 0 & 0 \\
0 & 0 & 0 & 0 & 0 & 1 & 0 & 0\\
0 & 0 & 0 & 1 & 0 & 0 & 0 & 0\\
0 & 0 & 0 & 1 \\
0& 0 & 1 & 0  \\
\end{pmatrix};~~~
\begin{pmatrix}
1 & 0 & 0 & 0 \\
0 & 0 & 0 & 0 & 0 & 1 & 0 & 0\\
0 & 0 & 1 & 0 & 0 & 0 & 0 & 0\\
0 & 0 & 0 & 1 \\
0 & 0 & 1 & 0  \\
\end{pmatrix};~~~
\begin{pmatrix}
0 & 1 & 0 & 0 \\
0 & 0 & 0 & 0  & 1 & 0 & 0 & 0\\
0 & 0 & 0 & 1 & 0 & 0 & 0 & 0\\
0 & 0 & 1 & 0 \\
0& 0 & 1 & 0  \\
\end{pmatrix}; \\\\
\begin{pmatrix}
0 & 0 & 1 & 0 \\
0 & 0 & 0 & 0 & 0 & 0 & 1 & 0\\
1 & 0 & 0 & 0 & 0 & 0 & 0 & 0\\
1 & 0 & 0 & 0 \\
0 & 0 & 1 & 0  \\
\end{pmatrix};~~~
\begin{pmatrix}
0 & 0 & 0 & 1 \\
0 & 0 & 0 & 0 & 0 & 0 & 0 & 1\\
0 & 1 & 0 & 0 & 0 & 0 & 0 & 0\\
0 & 1 & 0 & 0 \\
0 & 0 & 1 & 0  \\
\end{pmatrix};~~~
\begin{pmatrix}
0 & 0 & 1 & 0 \\
0 & 0 & 0 & 0 & 0 & 0 & 0 & 1\\
1 & 0 & 0 & 0 & 0 & 0 & 0 & 0\\
0 & 1 & 0 & 0 \\
0 & 0 & 1 & 0  \\
\end{pmatrix};~~~
\begin{pmatrix}
0 & 0 & 0 & 1 \\
0 & 0 & 0 & 0 & 0 & 0 & 1 & 0\\
0 & 1 & 0 & 0 & 0 & 0 & 0 & 0\\
1 & 0 & 0 & 0 \\
0 & 0 & 1 & 0  \\
\end{pmatrix},
\end{split}
\end{aligned}
}\]
respectively, the $16$ deterministic boxes for $de=01$ in $P^{(\alpha\beta\gamma\epsilon)(de)}_D$ are given as follows:
\[\Scale[0.8]{
\begin{aligned}\label{D01}
\begin{split}
\begin{pmatrix}
1 & 0 & 0 & 0 \\
1 & 0 & 0 & 0 & 0 & 0 & 0 & 0 \\
0 & 0 & 0 & 0 & 1 & 0 & 0 & 0 \\
1 & 0 & 0 & 0 \\
0 & 1 & 0 & 0  \\
\end{pmatrix};~~~
\begin{pmatrix}
0 & 1 & 0 & 0 \\
0 & 1 & 0 & 0 & 0 & 0 & 0 & 0\\
0 & 0 & 0 & 0 & 0 & 1 & 0 & 0\\
0 & 1 & 0 & 0 \\
0& 1 & 0 & 0  \\
\end{pmatrix};~~~
\begin{pmatrix}
1 & 0 & 0 & 0 \\
0 & 1 & 0 & 0 & 0 & 0 & 0 & 0 \\
0 & 0 & 0 & 0 & 1 & 0 & 0 & 0 \\
0 & 1 & 0 & 0 \\
0 & 1 & 0 & 0  \\
\end{pmatrix};~~~
\begin{pmatrix}
0 & 1 & 0 & 0 \\
1 & 0 & 0 & 0 & 0 & 0 & 0 & 0\\
0 & 0 & 0 & 0  & 0 & 1 & 0 & 0 \\
1 & 0 & 0 & 0\\
0& 1 & 0 & 0  \\
\end{pmatrix};\\\\
\begin{pmatrix}
0 & 0 & 1 & 0 \\
0 & 0 & 1 & 0 & 0 & 0 & 0 & 0\\
0 & 0 & 0 & 0 & 0 & 0 & 1 & 0 \\
0 & 0 & 1 & 0 \\
0 & 1 & 0 & 0  \\
\end{pmatrix};~~~
\begin{pmatrix}
0 & 0 & 0 & 1 \\
0 & 0 & 0 & 1 & 0 & 0 & 0 & 0 \\
0 & 0 & 0 & 0 & 0 & 0 & 0 & 1 \\
0 & 0 & 0 & 1 \\
0 & 1 & 0 & 0  \\
\end{pmatrix};~~~
\begin{pmatrix}
0 & 0 & 1 & 0 \\
0 & 0 & 0 & 1 & 0 & 0 & 0 & 0 \\
0 & 0 & 0 & 0 & 0 & 0 & 1 & 0 \\
0 & 0 & 0 & 1 \\
0 & 1 & 0 & 0  \\
\end{pmatrix};~~~
\begin{pmatrix}
0 & 0 & 0 & 1 \\
0 & 0 & 1 & 0 & 0 & 0 & 0 & 0 \\
0 & 0 & 0 & 0 & 0 & 0 & 0 & 1\\
0 & 0 & 1 & 0 \\
0 & 1 & 0 & 0  \\
\end{pmatrix};\\\\
\begin{pmatrix}
1 & 0 & 0 & 0 \\
1 & 0 & 0 & 0 & 0 & 0 & 0 & 0 \\
0 & 0 & 0 & 0 & 0 & 0 & 1 & 0 \\
0 & 0 & 1 & 0 \\
0 & 1 & 0 & 0  \\
\end{pmatrix};~~~
\begin{pmatrix}
0 & 1 & 0 & 0 \\
0 & 1 & 0 & 0 & 0 & 0 & 0 & 0\\
0 & 0 & 0 & 0 & 0 & 0 & 0 & 1\\
0 & 0 & 0 & 1 \\
0& 1 & 0 & 0  \\
\end{pmatrix};~~~
\begin{pmatrix}
1 & 0 & 0 & 0 \\
0 & 1 & 0 & 0 & 0 & 0 & 0 & 0\\
0 & 0 & 0 & 0 & 0 & 0 & 1 & 0\\
0 & 0 & 0 & 1 \\
0 & 1 & 0 & 0  \\
\end{pmatrix};~~~
\begin{pmatrix}
0 & 1 & 0 & 0 \\
1 & 0 & 0 & 0 & 0 & 0 & 0 & 0\\
0 & 0 & 0 & 0 & 0 & 0 & 0 & 1\\
0 & 0 & 1 & 0 \\
0& 1 & 0 & 0  \\
\end{pmatrix}; \\\\
\begin{pmatrix}
0 & 0 & 1 & 0 \\
0 & 0 & 1 & 0 & 0 & 0 & 0 & 0\\
0 & 0 & 0 & 0  & 1 & 0 & 0 & 0\\
1 & 0 & 0 & 0\\
0 & 1 & 0 & 0  \\
\end{pmatrix};~~~
\begin{pmatrix}
0 & 0 & 0 & 1 \\
0 & 0 & 0 & 1 & 0 & 0 & 0 & 0\\
0 & 0 & 0 & 0 & 0 & 1 & 0 & 0\\
0 & 1 & 0 & 0 \\
0 & 1 & 0 & 0  \\
\end{pmatrix};~~~
\begin{pmatrix}
0 & 0 & 1 & 0 \\
0 & 0 & 0 & 1 & 0 & 0 & 0 & 0\\
0 & 0 & 0 & 0 & 1 & 0 & 0 & 0\\
0 & 1 & 0 & 0 \\
0 & 1 & 0 & 0  \\
\end{pmatrix};~~~
\begin{pmatrix}
0 & 0 & 0 & 1 \\
0 & 0 & 1 & 0 & 0 & 0 & 0 & 0\\
0 & 0 & 0 & 0 & 0 & 1 & 0 & 0\\
1 & 0 & 0 & 0 \\
0 & 1 & 0 & 0  \\
\end{pmatrix},
\end{split}
\end{aligned}
}\]
and, finally, the $16$ deterministic boxes for $de=11$ in $P^{(\alpha\beta\gamma\epsilon)(de)}_D$ are given as follows:
\[\Scale[0.8]{
\begin{aligned}\label{D11}
\begin{split}
\begin{pmatrix}
1 & 0 & 0 & 0 \\
0 & 0 & 0 & 0 & 1 & 0 & 0 & 0\\
0 & 0 & 0 & 0  & 1 & 0 & 0 & 0\\
1 & 0 & 0 & 0\\
0 & 0 & 0 & 1  \\
\end{pmatrix};~~~
\begin{pmatrix}
0 & 1 & 0 & 0 \\
0 & 0 & 0 & 0  & 0 & 1 & 0 & 0\\
0 & 0 & 0 & 0 & 0 & 1 & 0 & 0\\
0 & 1 & 0 & 0 \\
0& 0 & 0 & 1  \\
\end{pmatrix};~~~
\begin{pmatrix}
1 & 0 & 0 & 0 \\
0 & 0 & 0 & 0 & 0 & 1 & 0 & 0\\
0 & 0 & 0 & 0 & 1 & 0 & 0 & 0\\
0 & 1 & 0 & 0 \\
0 & 0 & 0 & 1  \\
\end{pmatrix};~~~
\begin{pmatrix}
0 & 1 & 0 & 0 \\
0 & 0 & 0 & 0 & 1 & 0 & 0 & 0\\
0 & 0 & 0 & 0 & 0 & 1 & 0 & 0\\
1 & 0 & 0 & 0 \\
0& 0 & 0 & 1  \\
\end{pmatrix};\\\\
\begin{pmatrix}
0 & 0 & 1 & 0 \\
0 & 0 & 0 & 0 & 0 & 0 & 1 & 0\\
0 & 0 & 0 & 0 & 0 & 0 & 1 & 0\\
0 & 0 & 1 & 0 \\
0 & 0 & 0 & 1  \\
\end{pmatrix};~~~
\begin{pmatrix}
0 & 0 & 0 & 1 \\
0 & 0 & 0 & 0 & 0 & 0 & 0 & 1\\
0 & 0 & 0 & 0 & 0 & 0 & 0 & 1\\
0 & 0 & 0 & 1 \\
0 & 0 & 0 & 1  \\
\end{pmatrix};~~~
\begin{pmatrix}
0 & 0 & 1 & 0 \\
0 & 0 & 0 & 0 & 0 & 0 & 0 & 1\\
0 & 0 & 0 & 0 & 0 & 0 & 1 & 0\\
0 & 0 & 0 & 1 \\
0 & 0 & 0 & 1  \\
\end{pmatrix};~~~
\begin{pmatrix}
0 & 0 & 0 & 1 \\
0 & 0 & 0 & 0 & 0 & 0 & 1 & 0\\
0 & 0 & 0 & 0 & 0 & 0 & 0 & 1\\
0 & 0 & 1 & 0 \\
0 & 0 & 0 & 1  \\
\end{pmatrix};\\\\
\begin{pmatrix}
1 & 0 & 0 & 0 \\
0 & 0 & 0 & 0 & 1 & 0 & 0 & 0\\
0 & 0 & 0 & 0 & 0 & 0 & 1 & 0\\
0 & 0 & 1 & 0 \\
0 & 0 & 0 & 1  \\
\end{pmatrix};~~~
\begin{pmatrix}
0 & 1 & 0 & 0 \\
0 & 0 & 0 & 0 & 0 & 1 & 0 & 0\\
0 & 0 & 0 & 0 & 0 & 0 & 0 & 1\\
0 & 0 & 0 & 1 \\
0& 0 & 0 & 1  \\
\end{pmatrix};~~~
\begin{pmatrix}
1 & 0 & 0 & 0 \\
0 & 0 & 0 & 0  & 0 & 1 & 0 & 0\\
0 & 0 & 0 & 0  & 0 & 0 & 1 & 0\\
0 & 0 & 0 & 1 \\
0 & 0 & 0 & 1  \\
\end{pmatrix};~~~
\begin{pmatrix}
0 & 1 & 0 & 0 \\
0 & 0 & 0 & 0  & 1 & 0 & 0 & 0\\
0 & 0 & 0 & 0 & 0 & 0 & 0 & 1\\
0 & 0 & 1 & 0 \\
0& 0 & 0 & 1  \\
\end{pmatrix}; \\\\
\begin{pmatrix}
0 & 0 & 1 & 0 \\
0 & 0 & 0 & 0 & 0 & 0 & 1 & 0\\
0 & 0 & 0 & 0 & 1 & 0 & 0 & 0\\
1 & 0 & 0 & 0 \\
0 & 0 & 0 & 1  \\
\end{pmatrix};~~~
\begin{pmatrix}
0 & 0 & 0 & 1 \\
0 & 0 & 0 & 0 & 0 & 0 & 0 & 1\\
0 & 0 & 0 & 0 & 0 & 1 & 0 & 0\\
0 & 1 & 0 & 0 \\
0 & 0 & 0 & 1  \\
\end{pmatrix};~~~
\begin{pmatrix}
0 & 0 & 1 & 0 \\
0 & 0 & 0 & 0 & 0 & 0 & 0 & 1\\
0 & 0 & 0 & 0 & 1 & 0 & 0 & 0\\
0 & 1 & 0 & 0 \\
0 & 0 & 0 & 1  \\
\end{pmatrix};~~~
\begin{pmatrix}
0 & 0 & 0 & 1 \\
0 & 0 & 0 & 0 & 0 & 0 & 1 & 0\\
0 & 0 & 0 & 0 & 0 & 1 & 0 & 0\\
1 & 0 & 0 & 0 \\
0 & 0 & 0 & 1  \\
\end{pmatrix},
\end{split}
\end{aligned}
}\]
respectively.

\section{Decomposition of the box (\ref{nPbncE}) in terms of noncontextual deterministic boxes}\label{CPnPbncE}
There are $16$ deterministic boxes that satisfy the zero probabilities of the box 
(\ref{nPbncE}). Any convex mixture of these boxes given by
\begin{align}
P&=p_1P^{0000(00)}_D+p_2P^{0101(00)}_D+p_3P^{1011(00)}_D+p_4P^{1110(00)}_D  \nonumber\\
&+p_5P^{0010(10)}_D+p_6P^{0111(10)}_D+p_7P^{1001(10)}_D+p_8P^{1100(10)}_D  \nonumber\\
&+p_9P^{0011(01)}_D+p_{10}P^{0110(01)}_D+p_{11}P^{1000(01)}_D+p_{12}P^{1101(01)}_D \nonumber\\ 
&+p_{13}P^{0001(11)}_D+p_{14}P^{0100(11)}_D+p_{15}P^{1010(11)}_D+p_{16}P^{1111(11)}_D,
\end{align}
has the zero probabilities of the box (\ref{nPbncE}).
To check the noncontextuality of the box (\ref{nPbncE}), it suffices to show if there are weights, $0\le p_i \le 1$, with $i=1,2\cdots 16$ and  $\sum_ip_i=1$, in the above decomposition to produce all other nonzero probabilities of 
the box (\ref{nPbncE}). Using such a decomposition, it has been checked that
the box admits the following decomposition:
\begin{align} \label{DBDnP}
P(c_i|C_i) &= \frac{1}{8}
P^{0000(00)}_D+\frac{1}{8}P^{0101(00)}_D+\frac{1}{24}P^{1011(00)}_D+\frac{1}{24}P^{1110(00)}_D  \nonumber \\
&+  \frac{1}{24}
P^{0010(10)}_D+\frac{1}{24}P^{0111(10)}_D+\frac{1}{24}P^{1001(10)}_D+\frac{1}{24}P^{1100(10)}_D  \nonumber \\
&+  \frac{1}{24}
P^{0011(01)}_D+\frac{1}{24}P^{0110(01)}_D+\frac{1}{24}P^{1000(01)}_D+\frac{1}{24}P^{1101(01)}_D  \nonumber \\
&+  \frac{1}{24}
P^{0001(11)}_D+\frac{1}{24}P^{0100(11)}_D+\frac{1}{8}P^{1010(11)}_D+\frac{1}{8}P^{1111(11)}_D.   
\end{align}
It then follows that the box is noncontextual as it admits a decomposition in terms of the noncontextual deterministic boxes.

\section{Examples of supernoncontextuality without superlocality}  \label{Exsncwsl}
The following separable state has a rank of $2$:
\begin{equation}\label{rank2}
\rho^{\texttt{rank}-2}_{\texttt{QQ}}=\frac{1}{2} \left( \ketbra{00}{00} +\ketbra{++}{++}  \right),
\end{equation}
where $\ket{+}$ is the $+1$ eigenstate of the Pauli observable $X_2$,
does not have superlocality as its cardinality of separable decomposition is $2$ \cite{JKC+24}.
Consider the box arising from the two-qubit state (\ref{rank2}) for the observables given by Eq. (\ref{ObsP}). This box is given by,
\begin{align}\label{sncbd}
P(c_i|C_i)&=
\begin{pmatrix}
\frac{5}{8} & \frac{1}{8} & \frac{1}{8} & \frac{1}{8} \\ \\[0.05cm]
\frac{1}{2} & 0 & 0 & 0 & 0 & \frac{1}{4} & \frac{1}{4} & 0 \\ \\[0.05cm]
\frac{1}{2} & 0 & 0 & 0 & 0 & \frac{1}{4} & \frac{1}{4} & 0 \\ \\[0.05cm]
\frac{5}{8} & \frac{1}{8} & \frac{1}{8} & \frac{1}{8}  \\\\[0.05cm]
\frac{1}{4} & \frac{1}{4} & \frac{1}{4} & \frac{1}{4} \\
\end{pmatrix}.
\end{align}
The above box is noncontextual as it has a noncontextual deterministic model as 
follows:
\begin{align} \label{disls}
P&=\frac{1}{8}\Big( 2P^{(0000)(00)} + P^{(0010)(10)} +P^{(0011)(01)} \\ \nonumber
&+P^{(1000)(01)}+P^{(1010)(11)}+P^{(1100)(10)} \\ \nonumber
&+P^{(1111)(11)}  \Big).
\end{align}

For the box (\ref{sncbd}), from the  decomposition of the box in terms of the deterministic boxes as given by Eq. (\ref{disls}), we have obtained the following noncontextual hidden variable model:
\begin{align}
P(c_i|C_i)&=\sum^{d_\lambda-1}_{\lambda=0} p(\lambda) P(c_i|C_i, \lambda),
\end{align}
with the dimension of the hidden variable $d_\lambda =7$.  Here $\sum_{\lambda} p(\lambda) = 1$, $P(c|C_i,\lambda)$ = $\{p(c|C_i,\lambda)\}_{c,C_i}$  denotes a  non-contextual box for the given $\lambda$ as in Eq. (\ref{disls}) ($\lambda$ occurs with probability $p(\lambda)=1/4$ for $\lambda=0,1$ and $p(\lambda)=1/8$, otherwise).  It has been checked that 
the simulation of the box using the noncontextual deterministic boxes requires 
the above-mentioned seven deterministic boxes. This implies that the box has the minimal $d_\lambda=7$. Therefore, we have obtained the following observation.
\begin{observation}
    The  box $P(c_i|C_i)$ given by Eq. (\ref{sncbd}) is supernoncontextual.  In other words, the box requires a hidden variable of dimension $d_\lambda = 7$ to be simulated using a noncontextual hidden variable model, 
on the other hand, it can be simulated quantum mechanically using a quantum system with global Hilbert space dimension $d^g_Q=4$.
\end{observation}

We now proceed to provide our next example of supernoncontextuality  without superlocality. Before that, we provide an example of a noncontextual box that arises from a classically-correlated state for incompatible measurements given by Eq. (\ref{ObsP}), but has no supernoncontextuality.  
Let us consider the following noncontextual box:
\begin{align}\label{cc2q}
P(c_i|C_i)&=
\begin{pmatrix}
\frac{1}{2} & 0 & 0 & \frac{1}{2} \\ \\[0.05cm]
\frac{1}{4} & 0 & 0 & \frac{1}{4} & 0 & \frac{1}{4} & \frac{1}{4} & 0\\ \\[0.05cm]
\frac{1}{4} & 0 & 0 & \frac{1}{4} & 0 & \frac{1}{4} & \frac{1}{4} & 0\\ \\[0.05cm]
\frac{1}{4} & \frac{1}{4} & \frac{1}{4} & \frac{1}{4} \\ \\[0.05cm]
\frac{1}{4} & \frac{1}{4} & \frac{1}{4} & \frac{1}{4}  \\
\end{pmatrix}
\end{align}
which can be produced using the CC state $\rho_{\texttt{CC}}$ given by \begin{equation}\label{ccs}
\rho_{\texttt{CC}}=\frac{1}{2} (\ketbra{00}{00}+ \ketbra{11}{11}).
\end{equation}
for the observables given by Eq. (\ref{ObsP}).
The above box (\ref{cc2q}) can be reproduced using a
noncontextual hidden variable model with $d_\lambda=4$ as follows:
\begin{align}
P(c_i|C_i)&=\sum^{3}_{\lambda=0} p(\lambda) P(c_i|C_i, \lambda).
\end{align}
Here $\sum_{\lambda} p(\lambda) = 1$, $P(c|C_i,\lambda)$ = $\{p(c_i|C_i,\lambda)\}_{c_i,C_i}$  denotes a deterministic non-contextual box for the given $\lambda$ given by 
$P(c_i|C_i,0)=P^{(0000)(00)}_D$, $P(c_i|C_i,1)=P^{(0111)(10)}_D$,
$P(c_i|C_i,2)=P^{(1101)(01)}_D$ and $P(c_i|C_i,3)=P^{(1010)(11)}_D$
are the deterministic boxes given in Appendix \ref{matrixNota}.
($\lambda$ occurs with probability $p(\lambda)=\frac{1}{4}$, $\forall$ $\lambda$).

In the above example of a noncontextual box which has the minimal $d_\lambda=4$, correlations are present in the three 
contexts $C_0$, $C_1$ and $C_2$. We now proceed to provide another example of
a noncontextual box arising from the classically-correlated state (\ref{ccs}), but having supernoncontextuality.
In this example, the number of contexts having correlations is increased to the four contexts $C_0$, $C_1$, $C_2$ and $C_3$. This is achieved by invoking the following choice of observables:
\begin{small}
\begin{align}\label{Obscc}
\begin{split}
&A_0=\frac{Z_2+X_2}{\sqrt{2}} \otimes \openone_2, \quad B_0= \openone_2 \otimes \frac{Z_2+X_2}{\sqrt{2}},  \\
&B_1= \openone_2 \otimes \frac{Z_2-X_2}{\sqrt{2}}, \quad A_1=\frac{Z_2-X_2}{\sqrt{2}}\otimes \openone_2, \\
&D=\frac{Z_2+X_2}{\sqrt{2}} \otimes \frac{Z_2-X_2}{\sqrt{2}}, \quad E=  \frac{Z_2-X_2}{\sqrt{2}} \otimes  \frac{Z_2+X_2}{\sqrt{2}}.  
\end{split}
\end{align}
\end{small}
For this choice of observables, we obtain the following noncontextual box:
\begin{align}\label{cc2qhd}
P(c_i|C_i)&=
\begin{pmatrix}
\frac{3}{8} & \frac{1}{8} & \frac{1}{8} & \frac{3}{8} \\ \\[0.05cm]
\frac{3}{8} & 0 & 0 & \frac{3}{8} & 0 & \frac{1}{8} & \frac{1}{8} & 0\\ \\[0.05cm]
\frac{3}{8} & 0 & 0 & \frac{3}{8} & 0 & \frac{1}{8} & \frac{1}{8} & 0\\ \\[0.05cm]
\frac{3}{8} & \frac{1}{8} & \frac{1}{8} & \frac{3}{8} \\ \\[0.05cm]
\frac{1}{2} & \frac{1}{4} & \frac{1}{4} &  0  \\
\end{pmatrix}
\end{align}
The above box can be decomposed in terms of the noncontextual deterministic boxes 
in more than one way. Out of these models, we present the following 
model which provides a noncontextual hidden variable model with  $d_\lambda$ minimized. The box (\ref{cc2qhd}) has the following noncontextual deterministic hidden variable model:
\begin{align}
P(c_i|C_i)&=\sum^{d_\lambda-1}_{\lambda=0} p(\lambda) P(c_i|C_i, \lambda),
\end{align}
with dimension of the hidden variable $d_\lambda =6$.  Here $\sum_{\lambda} p(\lambda) = 1$, $P(c|C_i,\lambda)$ := $\{p(c_i|C_i,\lambda)\}_{c_i,C_i}$  denotes a deterministic non-contextual box for the given $\lambda$ given by 
$P(c_i|C_i,0)=P^{(0000)(00)}_D$, $P(c_i|C_i,1)=P^{(0101)(00)}_D$,
$P(c_i|C_i,2)=P^{(0010)(10)}_D$, $P(c_i|C_i,3)=P^{(0111)(10)}_D$, 
$P(c_i|C_i,4)=P^{(0011)(01)}_D$ and $P(c_i|C_i,5)=P^{(0110)(01)}_D$,
are the deterministic boxes given in Appendix \ref{matrixNota}
($\lambda$ occurs with probability $p(\lambda)=\frac{1}{4}$, for $\lambda=0,1$ and $p(\lambda)=\frac{1}{8}$, otherwise).  We have now obtained the following observation.
\begin{observation}
    The  box $P(c_i|C_i)$ given by Eq. (\ref{cc2qhd}) is supernoncontextual.  In other words, the box requires a hidden variable of dimension $d_\lambda = 6$ to be simulated using a noncontextual hidden variable model, 
on the other hand, it can be simulated quantum mechanically using a quantum system with lower global Hilbert space dimension $d^g_Q=4$. Supernoncontextuality of the box (\ref{cc2qhd}) as quantified by $d_\lambda$
   is lower than that of the box (\ref{sncbd}).
\end{observation}

\bibliographystyle{quantum}
\bibliography{context}

\end{document}